%% file: main.tex
\documentclass[conference,a4paper]{IEEEtran}
\IEEEoverridecommandlockouts
\usepackage{mleftright}       
\mleftright                   

\usepackage{graphicx}         
\usepackage{booktabs} 

\usepackage[utf8]{inputenc}
\usepackage{amsmath,amssymb,amsfonts}
\usepackage{mathtools}
\usepackage{amsthm}
\usepackage{algorithmic}
\usepackage{textcomp}
\usepackage{tikz}
\usepackage{caption}
\usepackage{cuted}
\usepackage{pgfgantt}
\usepackage{pdflscape}
\usepackage{pst-plot}
\usepackage{comment} 
\usepackage{cases}
\usepackage{nccfoots}
\usepackage{lineno,hyperref}
\usetikzlibrary{spy}
\usetikzlibrary{positioning,calc}
\usetikzlibrary{decorations.pathmorphing,calc,shapes,shapes.geometric,patterns}
\usetikzlibrary{shapes.multipart}
\usepackage{xfrac}
\usepackage{colortbl}
\usepackage{cancel} 
\usetikzlibrary{arrows,positioning,calc,intersections}
\usetikzlibrary{datavisualization.formats.functions}
\def\BibTeX{{\rm B\kern-.05em{\sc i\kern-.025em b}\kern-.08em
    T\kern-.1667em\lower.7ex\hbox{E}\kern-.125emX}}
    
\usepackage{romannum}
\usepackage{pgfplots}
\usepgfplotslibrary{fillbetween}
\usetikzlibrary{arrows, decorations.markings}
\usetikzlibrary{arrows.meta}

\newtheorem{theorem}{Theorem}
\newtheorem*{theorem*}{Theorem}

\newtheorem{definition}[theorem]{Definition}

\newtheorem{corollary}[theorem]{Corollary}

\newcommand{\setx}{\ensuremath{\mathcal{X}}}
\newcommand{\sety}{\ensuremath{\mathcal{Y}}}

\newcommand{\setd}{\ensuremath{\mathcal{D}}}

\newcommand{\sett}{\ensuremath{\mathcal{T}}}

\newcommand{\sets}{\ensuremath{\mathcal{S}}}

\newcommand{\up}{u^\prime}
\newcommand{\upp}{u^{\prime\prime}}
\newcommand{\Mp}{M^\prime}
\newcommand{\Mpp}{M^{\prime\prime}}
\newcommand{\dc}{\mathcal{D}}
\newcommand{\dcp}{\mathcal{D}^\prime}
\newcommand{\dcpp}{\mathcal{D}^{\prime\prime}}
\newcommand{\C}{\mathcal{C}}
\newcommand{\Cp}{\mathcal{C}^\prime}
\newcommand{\Cpp}{\mathcal{C}^{\prime\prime}}

\newcommand{\pmax}{P_{\text{max}}}
\newcommand{\pavg}{P_{\text{avg}}}

\DeclarePairedDelimiterX{\infdivx}[2]{(}{)}{%
  #1\;\delimsize\|\;#2%
}
\newcommand{\infdiv}{D\infdivx}
\DeclareMathOperator*{\pliminf}{p-liminf}
\DeclareMathOperator*{\plimsup}{p-limsup}
\begin{document}
\title{Identification Capacity of the Discrete-Time Poisson Channel} 



\author{
\IEEEauthorblockN{Wafa Labidi \IEEEauthorrefmark{1}\IEEEauthorrefmark{2}\IEEEauthorrefmark{4}, Christian Deppe\IEEEauthorrefmark{2}\IEEEauthorrefmark{4} and Holger Boche\IEEEauthorrefmark{1}\IEEEauthorrefmark{3}\IEEEauthorrefmark{4}\IEEEauthorrefmark{5}\IEEEauthorrefmark{6}}
\IEEEauthorblockA{\IEEEauthorrefmark{1}Technical University of Munich\\
\IEEEauthorrefmark{2}Technical University of Braunschweig\\
\IEEEauthorrefmark{3}Cyber Security in the Age of Large-Scale Adversaries–
Exzellenzcluster, Ruhr-Universit\"at Bochum, Germany\\
\IEEEauthorrefmark{4}BMBF Research Hub 6G-life, Germany\\
\IEEEauthorrefmark{5}{\color{black}{ Munich Center for Quantum Science and Technology (MCQST) }}\\
\IEEEauthorrefmark{6} {\color{black}{Munich Quantum Valley (MQV)}} \\
Email: wafa.labidi@tum.de, christian.deppe@tu-braunschweig.de, boche@tum.de}
}

\maketitle


\begin{abstract}
Numerous applications in the field of molecular communications (MC) such as healthcare systems are often event-driven. The conventional Shannon capacity may not be the appropriate metric for assessing performance in such cases. We propose the identification (ID) capacity as an alternative metric. Particularly, we consider randomized identification (RI) over the discrete-time Poisson channel (DTPC), which is typically used as a model for MC systems that utilize molecule-counting receivers. In the ID paradigm, the receiver's focus is not on decoding the message sent. However, he wants to determine whether a message of particular significance to him has been sent or not. In contrast to Shannon transmission codes, the size of ID codes for a Discrete Memoryless Channel (DMC) grows doubly exponentially fast with the blocklength, if randomized encoding is used. In this paper, we derive the capacity formula for RI over the DTPC subject to some peak and average power constraints. Furthermore, we analyze the case of state-dependent DTPC. 

\end{abstract}
\section{Introduction}
Molecular communication (MC) represents an emerging communication paradigm, using molecules or ions as the carriers of information \cite{nakano2012,farsad2016}. MC has arisen as a promising technique for detecting and pinpointing abnormality in diverse applications including biotechnology and medicine, such as drug delivery \cite{muller2004challenges}, cancer treatment \cite{jain1999transport} and health monitoring \cite{nakano2014molecular}. In particular, the use of MC for abnormality detection in medical applications has been extensively studied over the past decade \cite{AbnormalityDetection1,AbnormalityDetection2,AbnormalityDetection3}. In this context, by abnormality detection, we are referring to the detection of viruses, bacteria, infectious microorganisms, or tumors within the body, via either static \cite{AbnormalityDetection1,AbnormalityDetection2,AbnormalityDetection3} or mobile nanomachines \cite{nano1,nano2}. Once the abnormality has been detected, a suitable treatment is initiated using drug delivery, targeted therapy, nanosurgery, etc.
Information-theoretical analysis of diffusion-based MC was established in \cite{Mathfoundations,mathfoundation2}. For diffusion-based MC, the capacity limits of
molecular timing channels were investigated in \cite{capacity2} and lower and upper bounds on the corresponding
capacity were reported.


Numerous applications in MC such as healthcare systems are often event-driven, where the conventional Shannon capacity may not be the appropriate metric for assessing performance. Especially in situations involving event detection, where the receiver needs to make a dependable decision regarding the presence or absence of a particular event, the so-called identification (ID) capacity might be the performance metric of interest. The ID theory was first proposed by Ahlswede and Dueck in \cite{Idchannels, Idfeedback}. They were motivated by the work \cite{ja1985identification,yao1979some} on communication complexity. In the ID paradigm, the receiver’s focus is not on decoding the message sent. However, he wants to determine whether a message of particular significance to him has been sent or not. The sender does not know about the message that piques the receiver's interest. Otherwise, this would be a trivial task. Ahlswede and Dueck \cite{AhlDueck} proved that for Discrete Memoryless Channels (DMCs) the size of ID codes grows doubly exponentially $(2^{2^{nR}})$ fast with the blocklength, if randomized encoding is allowed. Later, ID codes have been developed \cite{verdu1993explicit, derebeyouglu2020performance, von2023identification}. If only deterministic encoding is allowed, the ID rate for a DMC is still more significant than the transmission rate in the exponential scale as shown in \cite{SPBD21pc,SPBD21f}. 
The number of ID messages that can be identified over Gaussian channels \cite{SPBD21pc} and Poisson channels \cite{SPBD21p,SV23isi} scales as ($n^{nR}$), if no randomization is used. Lower and upper bounds on the deterministic identification (DI) capacity over the discrete-time Poisson channel (DTPC) were established in \cite{SPBD21}. 
In MC, a sender, such as a cell, releases molecules into the channel at a specific rate (molecules/second) over a time interval. These molecules spread in the channel by diffusion and/or advection and can even be degraded in the channel by enzymatic reactions. Decoding assumes a counting receiver (nanomachines, cells) capable of counting the number of molecules received. The molecule production/release rate of the transmitter is of course limited. Assuming that the release propagation and the reception of individual molecules are statistically similar but independent, the received signal will follow Poisson statistics if the number of molecules released is large. Therefore, it makes sense to model such channels by a DTPC. The constraints on average and peak power take into account the limited molecule production/release rate of the sender \cite{jamali2019channel, gohari2016information, unterweger2018experimental}. Furthermore, the DTPC is of significant importance in the field of optical communications, where at the transmitter, a laser emits a stream of discrete photons and the receiver consists of a photodetector, which can discern the accurate arrival times of individual photons \cite{shamaiCapacityAchieving}.

In this work, we consider randomized identification (RI) over the DTPC. To the best of our knowledge, the RI capacity of the DTPC has not been studied so far. We derive the RI capacity of the DTPC under average and peak power constraints that account for the limited molecule production/release rate of the sender. We then extend the results to the DTPC with independent and identically distributed (i.i.d.) state sequences. The random channel state is available to neither the sender nor the receiver.  We show that, in this case, the RI capacity of the state-dependent DTPC coincides with its RI capacity.

\textit{Outline:} The remainder of the paper is structured as follows. In Section \ref{sec:Preli}, we introduce our system models, recall some basic definitions related to ID, and present the main results of the paper. In Section \ref{sec:Proof}, we establish the proof of the RI capacity formula of the DTPC and the state-dependent DTPC subject to some average and peak power constraints. Section \ref{sec: conclusion} concludes the paper.

\textit{Notation:}
$\mathbb{R}_0^+$ denotes the set of nonnegative real numbers; $\mathbb{Z}_0^+$ denotes the set of nonnegative integers;
the distribution of a RV $X$ is denoted by $P_X$; for a finite set $\setx$, we denote the set of probability distributions on $\setx$ by $\mathcal{P}(\setx)$ and by $|\setx|$ the cardinality of $\setx$; $\setx^c$ denotes the complement of $\setx$; if $X$ is a RV with distribution $P_X$, we denote the expectation of $X$ by $\mathbb{E}(X)$ and by ${\text{Var}}[X]$ the variance of $X$; if $X$ and $Y$ are two RVs with probability distributions $P_X$ and $P_Y$, the mutual information between $X$ and $Y$ is denoted by 
 $I(X;Y)$; the Kullback-Leibler divergence between $P_X$ and $P_Y$ is denoted by $\infdiv{P_X}{P_Y}$; 
all logarithms and information quantities are taken to the base $2$.
\section{System Models and Main Results} \label{sec:Preli}
Let a memoryless DTPC $(\setx, \sety, W(y|x))$ consisting of input alphabet $\setx \subset\mathbb{R}_0^+$, output alphabet $\sety\subset \mathbb{Z}_0^+$ and a pmf $W(y|x)$ on $\sety$, be given.
For $n$ channel uses, the transition probability law is given by
\begin{align*}
    W^n(y^n|x^n)&= \prod_{t=1}^n W(y_t|x_t) \\
    & = \prod_{t=1}^n \exp\big(-(x_t+\lambda_0)\big) \frac{(x_t+\lambda_0)^{y_t}}{{y_t}!}, 
\end{align*}
where $\lambda_0$ is some nonnegative constant, called dark current. The dark current $\lambda_0$, here, represents the non-ideality of the detector. The sequences $x^n=(x_1,x_2,\ldots,x_n) \in \setx^n$ and $y^n=(y_1,y_2,\ldots,y_n) \in \sety^n$ are the channel input and the channel output, respectively. 
 The peak and average power constraints on the input are 
\begin{align}
    & x_t\leq \pmax, \quad \forall t=1,\ldots,n, \label{eq:peakPower}\\
    & \frac{1}{n} \sum_{t=1}^n x_t \leq \pavg, \label{eq:averagePower}
\end{align}
 where $P_{\text{max}}, \ P_{\text{avg}}>0$ represent the values for peak power and average power constraints, respectively. 
 Let $\setx(P_{\text{max}},P_{\text{avg}})$ denote the set of all input symbols $x \in \setx$ such that \eqref{eq:peakPower} and \eqref{eq:averagePower} are satisfied.
 Let $C(W, \pmax,\pavg)$ denote the Shannon transmission capacity of the DTPC $W$ under peak and average power constraints $\pmax$ and $\pavg$ as described in \eqref{eq:peakPower} and \eqref{eq:averagePower}, respectively. Then  $C(W, \pmax,\pavg)$ is given by \cite{shamaiCapacityAchieving}
\begin{equation}
    C(W,\pmax,\pavg)  =\max_{\substack{P_X \in \mathcal{P}(\setx) \\ X \in \setx(P_{\text{max}},P_{\text{avg}})}} I(X;Y). \label{eq:ShannonCapacityDTPC}
\end{equation}
For the sake of notational simplicity, we denote the DTPC described above by $W$. \\
In the setting depicted in Fig. \ref{figSystem}, the sender  releases molecules (ID message $i \in \mathcal{N}:=\{1,2,\ldots,N\}$) into the DTPC $W$ at a specific rate over a time interval. The receiver, here, (nanomachines, cells) wants to check whether a specific pathological biomarker exists around the target or not.  Assuming that the release propagation and the reception of individual molecules are statistically similar but independent, the received signal follows Poisson statistics.

In the following, we define RI codes for the DTPC introduced above.
\begin{definition}
	An $(n,N,\lambda_1,\lambda_2)$ RI code with $\lambda_1+\lambda_2<1$ for the DTPC $W$ is a family of pairs 
	$\{(Q(\cdot|i),\setd_i), \quad   i=1,\ldots,N\}$ with 
	\begin{equation*}
	Q(\cdot|i) \in \mathcal{P}(\setx^n(P_{\text{max}},P_{\text{avg}})), \ \setd_i \subset \sety^n,\ \forall i=1,\ldots,N,
	\end{equation*}
	such that the errors of the first kind and the second kind are bounded as follows.
	\begin{align}
	 \mu_1^{(i)} &\triangleq \int_{x^n \in \setx^n(P_{\text{max}},P_{\text{avg}})} Q(x^n|i) \nonumber\\
  & \quad \cdot W^n(\setd_i^c|x^n) dx^n  \leq \lambda_1,  \forall i, \label{eq:firsterror}\\
\mu_2^{(i,j)} & \triangleq \int_{x^n \in \setx^n(P_{\text{max}},P_{\text{avg}})} Q(x^n|i) \nonumber\\
& \quad \cdot W^n(\setd_j|x^n) dx^n  \leq \lambda_2, \forall i\neq j, \label{eq:seconderror} 
	\end{align}
 where $\setx^n(P_{\text{max}},P_{\text{avg}})=\underbrace{\setx(P_{\text{max}},P_{\text{avg}})\times \ldots \times \setx(P_{\text{max}},P_{\text{avg}})}_{n \text{ times}}$.
 
\end{definition}
The error $\mu_1^{(i)}$ in \eqref{eq:firsterror} is called the error of the first kind, which is produced by channel noise and fits the same error definition as for a transmission code. In addition, we have another kind of error $\mu_2^{(i,j)}$ called the error of the second kind \eqref{eq:seconderror}, which results from the ID code construction. In contrast to transmission code, we permit a degree of overlap among the decoding sets in the case of ID codes. This results in both a double exponential increase in the rate (e.g., for DMCs) and introduces another kind of error as described in \eqref{eq:seconderror}. \\
In the following, we define achievable ID rate and ID capacity for our system model.
\begin{definition}
	\begin{enumerate}
		\item The rate $R$ of an $(n,N,\lambda_1,\lambda_2)$ RI code for the channel $W$ is $R=\frac{\log\log(N)}{n}$ bits.
		\item The ID rate $R$ for $W$ is said to be achievable if for $\lambda \in (0,\frac{1}{2})$ there exists an $n_0(\lambda)$, such that for all $n\geq n_0(\lambda)$ there exists an $(n,2^{2^{nR}},\lambda,\lambda)$ RI code for $W$.
		\item The RI capacity $C_{ID}(W,\pmax,\pavg)$ of the channel $W$ under the peak and average power constraints \eqref{eq:peakPower} and \eqref{eq:averagePower} is the supremum of all achievable rates.
	\end{enumerate}
\end{definition}
The following Theorem characterizes the RI capacity of the DTPC $W$ under peak and average power constraints $\pmax$ and $\pavg$, respectively.
\begin{theorem} \label{theorem:randomizedID}
 The RI capacity $C_{ID}(W,\pmax,\pavg)$ of the channel $W$ under peak and average power constraints $\pmax$ and $\pavg$, respectively, is given by
	\begin{align*}
	C_{ID}(W,\pmax,\pavg)&=C(W,\pmax,\pavg).
	\end{align*}
\end{theorem}
It is to be noted that the right-hand side of \eqref{eq:theorem} is a convex optimization problem. Therefore, Theorem \ref{theorem:randomizedID} provides the formula for computing the ID capacity of the DTPC as a function of the peak power and average power values $\pmax$ and $\pavg$, respectively. Many studies have been dedicated to investigating the properties of the capacity-achieving distribution of the DTPC. In the absence of input constraints, the transmission capacity of the DTPC is infinite. The DTPC was addressed in \cite{BarlettaShamai} when only subject to peak power constraint and
was shown that the support size is of an order between $\sqrt{\pmax}$ and $\pmax \ln^2 \pmax$. An Analytical formulation of the transmission capacity under only average power constraints is still an open problem. However, bounds and asymptotic behaviors for the DTPC in different scenarios were established e.g., \cite{lapidoth2008capacity}, \cite{lapidothWang}, \cite{aminian}, etc. It was shown in \cite{shamaiCapacityAchieving} that the capacity-achieving distribution for the DTPC under an average-power constraint and/or a peak
power constraint is discrete. However, an analytical formula for the capacity-achieving distribution of the DTPC remains an unresolved question. In some instances, it was demonstrated that the optimal input distribution for the DTPC is not even computable \cite{boche2023algorithmic}.
All these results can be extended and applied to the RI case.   
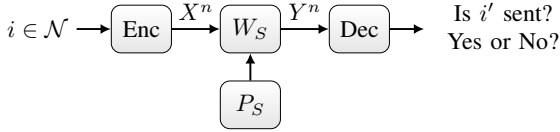
\begin{figure}
    \centering
    \input{figures/System}
    \caption{Discrete-time memoryless Poisson channel with random state}
    \label{figSystem}
\end{figure}
Let a DTPC with random state $(\setx\times \sets, W_S(y|x,s), \sety)$ consisting of an input alphabet $\setx \subset \mathbb{R}_0^+$, a discrete output alphabet $\sety \subset \mathbb{Z}_0^+$, a finite state set $\sets$ and a pmf $W(y|x,s)$ on $\sety$, be given. 
The channel is memoryless, i.e., the probability for a sequence $y^n \in \sety^n$ to be received if the input sequence $x^n \in \setx^n$ was sent and the sequence state is $s^n \in \sets^n$ is given by 
		\begin{equation*}
		W_S^n(y^n|x^n,s^n)=\prod_{t=1}^n W_S(y_t|x_t,s_t).
		\end{equation*}
		 The state sequence $(S_1,S_2,\ldots,S_n)$ is i.i.d. according to the distribution $P_S$. We assume that the input $X_t$ and the state $S_t$ are statistically independent for all $t\in\{1,2,\ldots,n\}$. We further assume that the input satisfies the peak power and average power constraints in \eqref{eq:peakPower} and \eqref{eq:averagePower}, respectively. We further assume that the channel state is  known to neither the sender nor the receiver. Now, let's define ID codes for the state-dependent DTPC $W_S$. 
\begin{definition}
	An $(n,N,\lambda_1,\lambda_2)$ RI code with $\lambda_1+\lambda_2<1$ for the channel $W_S$ is a family of pairs 
	$\{(Q(\cdot|i),\setd_i(s^n))_{s^n \in \sets^n}, \quad   i=1,\ldots,N\}$ with 
	\begin{equation*}
	Q(\cdot|i) \in \mathcal{P}(\setx^n(P_{\text{max}},P_{\text{avg}})), \ \setd_i(s^n) \subset \sety^n \end{equation*}
 for all $s^n \in \sets^n$ and for all $i=1,\ldots,N$
	such that the errors of the first kind and the second kind are bounded as follows.
	\begin{align*}
	&\sum_{s^n \in \sets^n} P_S^n(s^n) \int_{x^n \in \setx^n(P_{\text{max}},P_{\text{avg}})} Q(x^n|i) \\
 & \quad \cdot W_S^n(\setd_i(s^n)^c|x^n,s^n) dx^n  \leq \lambda_1,  \forall i,\\
	&\sum_{s^n \in \sets^n} P_S^n(s^n) \int_{x^n \in \setx^n(P_{\text{max}},P_{\text{avg}})} Q(x^n|i) \\
 & \quad \cdot W_S^n(\setd_j(s^n)|x^n,s^n) dx^n \leq \lambda_2, \forall i\neq j. 
	\end{align*}

\end{definition}
In the following, we define achievable ID rate and ID capacity for the state-dependent DTPC $W_S$.
\begin{definition}
	\begin{enumerate}
		\item The rate $R$ of an $(n,N,\lambda_1,\lambda_2)$ RI code for the channel $W_S$ is $R=\frac{\log\log(N)}{n}$ bits.
		\item The ID rate $R$ for $W_S$ is said to be achievable if for $\lambda \in (0,\frac{1}{2})$ there exists an $n_0(\lambda)$, such that for all $n\geq n_0(\lambda)$ there exists an $(n,2^{2^{nR}},\lambda,\lambda)$ randomized ID code for $W_S$.
		\item The RI capacity $C_{ID}(W_S,\pmax,\pavg)$ of the channel $W_S$ under peak and average power constraints \eqref{eq:peakPower} and \eqref{eq:averagePower} is the supremum of all achievable rates.
	\end{enumerate}
\end{definition}
The following corollary, a consequence of Theorem \ref{theorem:randomizedID}, outlines the RI capacity of the state-dependent DTPC $W_S$.
\begin{corollary} \label{corollary}
The RI capacity of the state-dependent DTPC $W_S$ under peak and average power constraints $\pmax$ and $\pavg$, respectively, is given by 
\begin{equation*}
    C_{ID}(W_S,\pmax,\pavg) =\max_{\substack{P_X \in \mathcal{P}(\setx) \\ X \in \setx(P_{\text{max}},P_{\text{avg}})}} I(X;Y).
\end{equation*}
\end{corollary}
\begin{proof}
    For the proof, we use the same idea as in \cite{NetworkITchap7} by averaging over the finite set of states. \\
Let $W^a(y|x) = \sum_{s \in \mathcal{S}} P_S(s) W_S(y|x,s)$ be the channel obtained by averaging the state-dependent DTPCs $W_S$ over the finite set of states $\mathcal{S}$.
Then, it can be readily verified that the Shannon capacity of the state-dependent DTPC $W_S$ when the state is available at neither the sender nor the receiver is given by
\begin{equation*}
    C(W_S, \pmax,\pavg) =\max_{\substack{P_X \in \mathcal{P}(\setx) \\ X \in \setx(P_{\text{max}},P_{\text{avg}})}} I(X;Y).
\end{equation*}
As the set of states $\sets$ is finite, we can show that the state-dependent DTPC $W_S$ satisfies the strong converse property as well. That means that the RI capacity of the state-dependent DTPC $W_S$ coincides with its Shannon transmission capacity. This completes the proof of Corollary \ref{corollary}.
\end{proof}
Now, we revisit the definitions of inf-mutual information rate and sup-mutual information rate defined in \cite{HanBook}. Those quantities will play a key role in proving Theorem \ref{theorem:randomizedID}.
We first recall the definition of $\liminf$ and $\limsup$ in probability.
\begin{definition}
For a given sequence $\{X^n\}_{n=1}^\infty$ of RVs, the $\liminf$ and $\limsup$ in probability are defined as
\begin{align*}
    \pliminf_{n\to \infty} X^n & := \sup \big\{ \beta\colon \lim_{n \to \infty} \Pr\{X^n< \beta\}=0 \big\} \\
    \plimsup_{n\to \infty} X^n & := \inf \big\{ \alpha\colon \lim_{n \to \infty} \Pr\{X^n>\alpha\}=0 \big\},
\end{align*}
respectively. We refer the reader to \cite{HanBook} for more details.
\end{definition}
The quantities mentioned above play a pivotal role in establishing information-spectrum methods.
\begin{definition}
    Let $\mathbf{X}=\{X^n\}_{n=1}^\infty$ be a general input process, where $X^n$ is an arbitrary RV with probability distribution $P_{X^n}$ on $\setx^n$. Let $\mathbf\{Y^n\}_{n=1}^\infty$ be the output of a channel $\mathbf{W}=\{W^n\}_{n=1}^\infty$ corresponding to the input $\mathbf{X}$. The output $Y^n$ is an arbitrary RV with probability distribution $P_{Y^n}$ on $\sety^n$ given by
    \begin{align*}
        P_{X^nY^n}(x^n,y^n) &= P_{X^n}(x^n) W^n(y^n|x^n).
    \end{align*}
    The spectral inf-mutual information rate and the spectral sup-mutual information rate defined in \cite{HanBook} are given by \eqref{eq:infMI} and \eqref{eq: supMI}, respectively.
   \begin{align}
  \underline{I}(\mathbf{X},\mathbf{Y}) &= \pliminf_{n\to \infty} \frac{1}{n} \frac{W^n(Y^n|X^n)}{P_{Y^n}(Y^n)}, \label{eq:infMI} \\
  \overline{I}(\mathbf{X},\mathbf{Y}) &= \plimsup_{n\to \infty} \frac{1}{n} \frac{W^n(Y^n|X^n)}{P_{Y^n}(Y^n)}, \label{eq: supMI}
   \end{align}
  where $\frac{1}{n} \frac{W^n(Y^n|X^n)}{P_{Y^n}(Y^n)}$ is called the mutual information density rate of $(\mathbf{X},\mathbf{Y})$ as defined in \cite{HanBook}.
\end{definition}

\section{Proof of Theorem \ref{theorem:randomizedID}} \label{sec:Proof}
In this section, we show that the RI capacity of the DTPC $W$ coincides with its transmission capacity. Indeed, it has been shown in \cite[Corollary 6.6.1]{HanBook} that if a channel $W$ satisfies the strong converse property under some cost constraint, then the corresponding RI capacity and transmission capacity coincide.   
In the following, we prove that the DTPC under peak power constraint $\pmax$ and average power constraint $\pavg$ satisfies the strong converse property. Let $\mathcal{U}(\pmax,\pavg)$ denote the set of all input processes $\mathbf{X}=\{X^n\}_{n=1}^\infty$ satisfying 
\begin{equation*}
    \Pr\{ X^n \in \setx^n(\pmax,\pavg)\}=1
\end{equation*}
for all $n=1,2, \ldots$. Then, as stated in \cite[Theorem 3.7.1]{HanBook}, the DTPC satisfies the strong converse property under cost constraints $\pmax$ and $\pavg$ if and only if 
\begin{align}
    \sup_{ \mathbf{X} \in \mathcal{U}(\pmax,\pavg)} \underline{I}(\mathbf{X},\mathbf{Y}) =  \sup_{ \mathbf{X} \in \mathcal{U}(\pmax,\pavg)}  \overline{I}(\mathbf{X},\mathbf{Y}).
    \label{eq:Strong}
\end{align}
By applying \cite[Theorem 3.6.1]{HanBook}, the transmission capacity of the DTPC $W$ under peak and average power constraints can be rewritten as follows:
\begin{equation*}
    C(W,\pmax,\pavg) = \sup_{ \mathbf{X} \in \mathcal{U}(\pmax,\pavg)} \underline{I}(\mathbf{X},\mathbf{Y}).
\end{equation*}
Thus, in order to develop the strong converse property for the DTPC $W$ under cost constraints $\pmax$ and $\pavg$, in view of \cite[Theorem 3.6.1]{HanBook} it suffices to show 
\begin{align}
    \sup_{ \mathbf{X} \in \mathcal{U}(\pmax,\pavg)}  \overline{I}(\mathbf{X},\mathbf{Y}) & \leq \sup_{\substack{P_X \in \mathcal{P}(\setx) \nonumber \\ X \in \setx(P_{\text{max}},P_{\text{avg}}) }} I(X;Y) \\
    & =\max_{\substack{P_X \in \mathcal{P}(\setx) \\ X \in \setx(P_{\text{max}},P_{\text{avg}})}} I(X;Y) \nonumber \\
    & = C(W,\pmax,\pavg).\label{eq:StrongConverseProperty}
\end{align}

Let $P_X \in \mathcal{P}\left(\setx(\pmax,\pavg)\right)$ an arbitray input distribution and $P_Y \in \mathcal{P}(\sety)$ the corresponding output distribution via the DTPC $W$. Then, for a fixed input sequence $x^n=(x_1,x_2,\ldots,x_n) \in \setx^n(\pmax,\pavg)$, we define $I({Y},x)$ as follows
\begin{align*}
    I({Y},x) & =\log \frac{W(Y|x)}{P_{{Y}}(Y)}.
\end{align*}
The conditional expectation of $I({Y},x)$ given $X=x$ is given by
\begin{align}
    \mathbb{E}[I({Y},x)]&= \sum_{y=0}^\infty W(y|x) \log \frac{W(y|x)}{P_{{Y}}(y)} \nonumber \\
    &= \infdiv{W(\cdot|x)}{P_{{Y}}}. \label{eq:expectationYi}
\end{align}
The mutual information induced by the input distribution $P_{X}$ is given by
\begin{align*}
    I(X,Y)&= \int_{0}^{\pmax} P_X(x)  \mathbb{E}[I({Y},x)] dx \\
    &=  \int_{0}^{\pmax} P_X(x)  \infdiv{W(\cdot|x)}{P_{{Y}}} dx.
\end{align*}

In the following, we define the Lagrangian function $L(\mu,x,P_{{X}})$ with Lagrange multiplier $\mu \geq 0$ as
\begin{align}
    L(\mu,x,P_{{X}}) &= I({X},{Y}) + \mu(x-\pavg)\nonumber\\ 
    & \quad -\infdiv{W(\cdot|x)}{P_{{Y}}}.
\end{align}
It has been shown in \cite{shamaiCapacityAchieving} and \cite[Theorem 5]{capacityAchievingDistribution} that Kuhn-Tucker conditions yield the following theorem.
\begin{theorem}[\cite{capacityAchievingDistribution}]
The distribution $P_{{X}} \in \mathcal{P}\left( \setx(\pmax,\pavg)\right)$ is capacity-achieving iff for some $\mu \geq 0$, the following conditions are satisfied 
\begin{align}
      L(\mu,x,P_{{X}})& \geq 0, \quad x \in [0,\pmax], \\
        L(\mu,x,P_{{X}}) & = 0, \quad  x \in \Phi_x,
\end{align}
where $\Phi_x$ is the support of $P_X$. \label{theorem:Lagrange}
\end{theorem}
It has been shown in \cite{shamaiCapacityAchieving} that the capacity-achieving input distribution is unique and discrete with a finite number of mass points for finite peak power and average power constraints. Let $P_{\bar{X}}$ denote the capacity-achieving input distribution of the DTPC under peak power and average power constraints $\pmax$ and $\pavg$, respectively. Therefore, as described in \cite{capacityAchievingDistribution} and w.l.o.g. the input distribution $P_{\bar{X}}$ is given by
\begin{align}
    P_{\bar{X}}(x) & = p_1 \delta(x-\bar{x}_1)+ p_2 \delta(x-\bar{x}_2)+\cdots + p_m\nonumber \\
    & \quad \cdot \delta(x-\bar{x}_m), \quad x \in \setx(\pmax,\pavg),
\end{align}
where $\delta(\cdot)$ denotes the Dirac impulse function, $\bar{\Phi}_x=\{\bar{x}_1,\bar{x}_2,\ldots,\bar{x}_m\}, \quad 0\leq \bar{x}_1< \bar{x}_2<\ldots<\bar{x}_m\leq \pmax $ is a finite input constellation and $\Phi_p=\{p_1,p_2,\ldots,p_m\}$ is the set of the corresponding  probability masses. Let $P_{\bar{Y}}$ denote the output distribution of the DTPC $W$ corresponding to the distribution $P_{\bar{X}}$. For all $y \in \sety$, we have
\begin{align}
    P_{\bar{Y}}(y) & = \int_{x \in \setx(\pmax,\pavg)} P_{\bar{X}}(x) W(y|x) dx \nonumber\\
    & = \int_{x \in \setx(\pmax,\pavg)} \big( \sum_{j=1}^m p_j \delta(x-\bar{x}_j) \big) W(y|x) dx \nonumber \\
    &= \sum_{j=1}^m p_j \int_{x \in \setx(\pmax,\pavg)} W(y|x) \delta(x-\bar{x}_j) dx \nonumber\\
    & = \sum_{j=1}^m p_j W(y|\bar{x}_j). \label{eq:pY}
\end{align}
Let $\bar{X}$ and $\bar{Y}$ be two RVs with probability distribution functions $P_{\bar{X}}$ and $P_{\bar{Y}}$ on $\setx(\pmax,\pavg)$ and $\sety$, respectively. In view of \eqref{eq:expectationYi}, the capacity $ C(W,\pmax,\pavg)$ of the DTPC $W$ can be rewritten as follows:
\begin{align*}
     &C(W,\pmax,\pavg)\\
      & \quad =\max_{\substack{P_X \in \mathcal{P}(\setx) \\ X \in \setx(P_{\text{max}},P_{\text{avg}})}} I(X;Y) \\
     &\quad = I(\bar{X},\bar{Y}) \\
     &\quad = \int_{x \in \setx(\pmax,\pavg)} P_{\bar{X}}(x) \sum_{y=0}^{\infty} W(y|x)(y|x_i) \log \frac{W(y|x)}{P_{\bar{Y}}(y)} dx \\
      & \quad = \int_{x \in \setx(\pmax,\pavg)} P_{\bar{X}}(x)\infdiv{W(\cdot|x)}{P_{\bar{Y}}} dx \\
     &= \quad \sum_{j=1}^m p_j \infdiv{W(\cdot|\bar{x}_j)}{P_{\bar{Y}}}.
\end{align*}
It holds for each blocklength $n$ that
\begin{align*}
    \frac{1}{n} \log \frac{W^n(Y^n|X^n)}{P_{\bar{Y}^n}(Y^n)} &= \frac{1}{n} \sum_{t=1}^n \log \frac{W(Y_t|X_t)}{P_{\bar{Y}}(Y_t)}.
\end{align*}
For fixed $x^n=(x_1,x_2,\ldots,x_n) \in \setx^n(\pmax,\pavg)$, we define
\begin{align*}
    I(\bar{Y}_t,x_t) &= \log \frac{W(Y_t|X_t)}{P_{\bar{Y}}(Y_t)}.
\end{align*}
 In addition, we have for all $t=1,2,\ldots, n$
\begin{align*}
    \mathbb{E}[I(\bar{Y}_t,x_t)] & = \infdiv{W(\cdot|x_t)}{P_{\bar{Y}}}.
\end{align*}
In order to show \eqref{eq:StrongConverseProperty}, we first prove the following inequality:
\begin{align}
    \mathbb{E}[\frac{1}{n} \sum_{t=1}^n I(\bar{Y}_i|x_i)] \leq C(W,\pmax,\pavg). \label{eq:StronConverse1}
\end{align}
It follows from Theorem \ref{theorem:Lagrange} that for $x \in [0,\pmax]$
\begin{align*}
     &L(\mu,x,P_{\bar{X}})\\
     &= I(\bar{X},\bar{Y}) + \mu(x-\pavg)-\infdiv{W(\cdot|x)}{P_{\bar{Y}}} \\
     &= C(W,\pmax,\pavg) + \mu(x-\pavg)-\infdiv{W(\cdot|x)}{P_{\bar{Y}}} \\
     & \geq 0. 
\end{align*}
As the sequence $x^n$ satisfies the peak power and average power constraints, we have
\begin{align}
  \mathbb{E}[I(\bar{Y},x)]= \infdiv{W(\cdot|x)}{P_{\bar{Y}}} & \leq  C(W,\pmax,\pavg). 
  \label{eq:expectationBound}
\end{align}
It yields that
\begin{align*}
\mathbb{E}[\frac{1}{n} \sum_{t=1}^n I(\bar{Y}_t,x_t)] & \leq C(W,\pmax,\pavg).
\end{align*}
This completes the proof of inequality \eqref{eq:StronConverse1}.\\
Next, we want to show that the variance of $I(\bar{Y}_t,x_t),\quad t=1,\ldots,n$ w.r.t. the conditional probability mass function $W(\cdot|x_t)$ for fixed $x_t \in \setx(\pmax,\pavg)$. For this purpose, we prove that $\mathbb{E}\big[(I(\bar{Y}_t,x_t))^2\big],\quad t=1,\ldots,n$ is finite. For notational simplicity, we will drop the index $t$. 
Let $\sety_1$ and $\sety_2$ be defined as
\begin{align*}
\sety_1 & = \{ y \in \sety \colon \frac{P_{\bar{Y}}(y)}{W(y|x)} > 1 \} ,\\
\sety_2 &=  \{ y \in \sety \colon \frac{P_{\bar{Y}}(y)}{W(y|x)} \leq 1 \}.
\end{align*}
We have
\begin{align}
 & \mathbb{E}\big[(I(\bar{Y},x))^2\big ] \\
 &= \sum_{y=0}^\infty W(y|x) \log^2 \big( \frac{W(y|x)}{P_{\bar{Y}}(y)}\big) \nonumber\\ 
 &=  \sum_{y \in \sety_1} W(y|x) \log^2 \big( \frac{W(y|x)}{P_{\bar{Y}}(y)}\big) \nonumber\\ 
 &\quad + \sum_{y\in \sety_2} W(y|x) \log^2 \big( \frac{W(y|x)}{P_{\bar{Y}}(y)}\big) \label{eq:sum}\\
\end{align}
We first establish an upper bound on the first term of the sum in \eqref{eq:sum}.
\begin{align}
& \sum_{y \in \sety_1} W(y|x) \log^2 \big( \frac{W(y|x)}{P_{\bar{Y}}(y)}\big) \nonumber\\
 & = \sum_{y\in \sety_1} W(y|x) \log^2 \big( \frac{P_{\bar{Y}}(y)}{W(y|x)}\big) \nonumber \\
 & \overset{(a)}{\leq}  \sum_{y\in \sety_1} W(y|x) \log \big( \frac{P_{\bar{Y}}(y)}{W(y|x)}\big) \cdot \log(e) \cdot \big( \frac{P_{\bar{Y}}(y)}{W(y|x)}\big) \nonumber \\
 &=  \log(e) \sum_{y\in \sety_1}  P_{\bar{Y}}(y) \log \big( \frac{P_{\bar{Y}}(y)}{W(y|x)}\big) \nonumber \\
 & \overset{(b)}{=} \log(e) \bigg(  \infdiv{P_{\bar{Y}}(\cdot)}{W(\cdot|x)} \nonumber \\
 & \quad - \sum_{y\in \sety2}  P_{\bar{Y}}(y) \log \big( \frac{P_{\bar{Y}}(y)}{W(y|x)}\big) \bigg) \nonumber \\
 &\overset{(c)}{ \leq} \log(e) \bigg(  \infdiv{P_{\bar{Y}}(\cdot)}{W(\cdot|x)} \nonumber \\
 & \quad - \sum_{y\in \sety2}  P_{\bar{Y}}(y)  \big( 1- \frac{W(y|x)}{P_{\bar{Y}}(y)}\big) \bigg) \nonumber \\
 & =  \log(e) \bigg(  \infdiv{P_{\bar{Y}}(\cdot)}{W(\cdot|x)} + \sum_{y\in \sety2} W(y|x)-P_{\bar{Y}}(y) \bigg) \nonumber \\
 & \overset{(d)}{=} \log(e) \bigg(  \infdiv{P_{\bar{Y}}(\cdot)}{W(\cdot|x)} + \frac{1}{2} d\left(P_{\bar{Y}}(\cdot),W(\cdot|x)\right) \bigg) \nonumber \\
 & \overset{(e)}{\leq}  \log(e) \bigg(  \infdiv{P_{\bar{Y}}(\cdot)}{W(\cdot|x)} +1 \bigg)\nonumber \\
 & \overset{(f)}{=}  \log(e) \bigg( \infdiv{\sum_{j=1}^m p_j W(\cdot| \bar{x}_j)}{\sum_{j=1}^m p_j W(\cdot|x)} +1 \bigg) \nonumber \\
 & \overset{(g)}{\leq} \log(e)  \bigg(\sum_{j=1}^m p_j  \infdiv{W(\cdot|\bar{x}_j}{W(\cdot|x)} + 1 \bigg) \nonumber \\
 & \overset{(h)}{\leq} \log(e) \bigg( \sum_{j=1}^m p_j \big( (\lambda+ \bar{x}_j) \log \frac{\lambda+ \bar{x}_j}{\lambda+x}  \nonumber\\
 & \quad (\bar{x}_j-x) \big) + 1 \bigg) \nonumber \\
 & \overset{(i)}{\leq} \log(e) \bigg( \sum_{j=1}^m p_j ( \lambda+ \pmax) \log \frac{\lambda+\pmax}{\lambda} + \pmax +1 \bigg) \nonumber \\
 & \overset{(j)}{\leq} \log(e) \bigg( \log(e)\frac{(\lambda+\pmax)^2}{\lambda}+\pmax +1 \bigg), \label{eq:firstterm}
\end{align}
where $(a)$ follows because $\log \big( \frac{P_{\bar{Y}}(y)}{W(y|x)}\big) >0$ and $\log(x) \leq (x-1)\log(e)$, $(b)$ follows from the definition of the Kullback-Leibler divergence, $(c)$ follows because $\log(x)\geq \log(e)(1-\frac{1}{x})$, $(d)$ follows from the definition of the variational distance $d\left(P_{\bar{Y}}(\cdot),W(\cdot|x)\right)$, 
$(e)$ follows because $0\leq d\left(P_{\bar{Y}}(\cdot),W(\cdot|x)\right) \leq 2$,
$(f)$ follows from \eqref{eq:pY}, $(g)$ follows from the convexity of the Kullback-Leibler divergence,
$(h)$ because $W(\cdot|\bar{x}_j), \quad j=1,\ldots,m$ and $W(\cdot|x)$ are Poisson distributed with mean $\lambda+\bar{x}_j$ and $\lambda+x$, respectively, 
$(i)$ follows because $0\leq x, \ \bar{x}_j\leq \pmax$ and $(j)$ follows because $\log(x) \leq (x-1)\log(e)$. \\
Now, we compute an upper bound on the second term of the sum in \eqref{eq:sum}. We have
\begin{align}
& \sum_{y \in \sety_2} W(y|x) \log^2 \big( \frac{W(y|x)}{P_{\bar{Y}}(y)}\big) \nonumber\\
&=  \sum_{y \in \sety_2} W(y|x) \big( \log(W(y|x)) - \log(P_{\bar{Y}}(y)) \big)^2 \nonumber \\
& \overset{(a)}{=}  \sum_{y \in \sety_2} W(y|x) \bigg( \log(W(y|x)) \nonumber \\
& \quad - \log\big(\sum_{j=1}^m p_j W(\cdot| \bar{x}_j)\big) \bigg)^2 \nonumber \\
&\overset{(b)}{\leq}  \sum_{y \in \sety_2} W(y|x) \bigg( \log(W(y|x)) \nonumber \\
& \quad - \sum_{j=1}^m p_j \log\big(W(\cdot| \bar{x}_j)\big)  \bigg)^2 \nonumber \\
& \overset{(c)}{=}  \sum_{y \in \sety_2} W(y|x) \bigg( \big(\log(x+\lambda) -  \sum_{j=1}^m p_j \log(x_j+\lambda) \big)y \nonumber \\
& \quad + \big(  \sum_{j=1}^m p_jx_j - x \big) \bigg)^2 \nonumber \\
&\overset{(d)}{\leq} \sum_{y \in \sety_2} W(y|x) \bigg( \alpha y + \beta \bigg)^2 \nonumber \\
&= \alpha^2\sum_{y \in \sety_2} W(y|x) y^2 + 2 \alpha  \beta \sum_{y \in \sety_2}  W(y|x) y \nonumber \\
& \quad + \beta^2 \sum_{y \in \sety_2} W(y|x) \nonumber \\
& \leq \alpha^2\mathbb{E}[Y^2|X=x]+  2 \alpha\nonumber  \beta \mathbb{E}[Y|X=x] + \beta^2(\lambda,\pmax) \nonumber \\
& \overset{(e)}{\leq}  \alpha^2 (\lambda+\pmax)^2 + 2 \alpha  \beta (\lambda+\pmax) +  \beta^2,\label{eq:secondterm}
\end{align}
where $d(\cdot,\cdot)$ denotes the total variational distance, $(a)$ follows \eqref{eq:pY}, $(b)$ follows from the convexity of $-\log(x)$, $(c)$ follows because $W(\cdot|\bar{x}_j), \quad j=1,\ldots,m$ and $W(\cdot|x)$ are Poisson distributed with mean $\lambda+\bar{x}_j$ and $\lambda+x$, respectively, $(d)$ follows from \eqref{eq:alpha} and \eqref{eq:beta} and $(e)$ follows because given the input $X=x$, the output $Y$ is Poisson distributed with mean $\lambda+x$ and $0<x\leq \pmax$. \\
Let $\alpha$ and $\beta$ be defined as follows:
\begin{align*}
    \alpha &= \log(1+\frac{\pmax}{\lambda} ) \\
    \beta & = \pmax.
\end{align*}
Then, it can be easily shown that 
\begin{align}
\log(x+\lambda) -  \sum_{j=1}^m p_j \log(x_j+\lambda) & \leq \alpha,
\label{eq:alpha}
\end{align}
and 
\begin{align}
    \sum_{j=1}^m p_jx_j - x & \leq \beta. \label{eq:beta}
\end{align}
It follows from \eqref{eq:firstterm} and \eqref{eq:secondterm} that
 \begin{align*}
    &{\text{Var}}[\frac{1}{n} \sum_{t=1}^n I(\bar{Y}_t,x_t)] \nonumber\\
    & \leq \log(e) \bigg( \log(e)\frac{(\lambda+\pmax)^2}{\lambda}+\pmax +1 \bigg) \nonumber \\
    & \quad +  \alpha^2 (\lambda+\pmax)^2 + 2 \alpha  \beta (\lambda+\pmax) +  \beta^2 .\\
\end{align*}
Let $\gamma(\lambda,\pmax)$ be defined as 
\begin{align*}
    \gamma(\lambda,\pmax) & = \log(e) \bigg( \log(e)\frac{(\lambda+\pmax)^2}{\lambda}+\pmax +1 \bigg) \\
    & \quad +  \alpha^2 (\lambda+\pmax)^2 + 2 \alpha  \beta (\lambda+\pmax) +  \beta^2.
\end{align*}
Therefore, Chebyshev's inequality implies
\begin{align}
    &\Pr\big\{ \frac{1}{n} \sum_{t=1}^n I(\bar{Y}_t,x_t) \geq C(W,\pmax,\pavg)+ \nu | X^n=x^n \big \} \nonumber\\
    & \quad \leq  \frac{\gamma(\lambda,\pmax)}{n}, \label{eq:tshebychev}
\end{align}
where $\nu>0$ is an arbitrary constant. As $\mathbf{X}=\{X^n\}_{n=1}^\infty \in \mathcal{U}(\pmax,\pavg)$ is assumed, \eqref{eq:tshebychev} holds for all realizations $x^n$ of $X^n$. Thus, we have
\begin{align*}
    &\Pr\big\{ \frac{1}{n} \sum_{t=1}^n I(\bar{Y}_t,X_t) \geq C(W,\pmax,\pavg)+ \nu  \big \} \\
    & \quad \leq  \frac{\gamma(\lambda,\pmax)}{n}.
\end{align*}
That means
\begin{align*}
    &\Pr\big\{ \frac{1}{n} \log \frac{W^n(Y^n|X^n)}{P_{\bar{Y}^n}(Y^n)} \geq C(W,\pmax,\pavg)+ \nu  \big \} \\ & \quad \leq  \frac{\gamma(\lambda,\pmax)}{n}. 
\end{align*}
We have
\begin{align*}
    &\lim_{n \to \infty} \Pr\big\{ \frac{1}{n} \log \frac{W^n(Y^n|X^n)}{P_{\bar{Y}^n}(Y^n)} \geq C(W,\pmax,\pavg)+ \nu  \big \} \\
    & \quad =0.
\end{align*}
Since $\nu>0$ is chosen arbitrarily, we have
\begin{align}
    \plimsup_{n \to \infty} \frac{1}{n} \log \frac{W^n(Y^n|X^n)}{P_{\bar{Y}^n}(Y^n)} & \leq  C(W,\pmax,\pavg). \label{eq:lastIneq}
\end{align}
It follows from \eqref{eq:lastIneq} and \cite[Lemma 3.2.1]{HanBook} that
\begin{align*}
    \plimsup_{n \to \infty} \frac{1}{n} \log \frac{W^n(Y^n|X^n)}{P_{{Y}^n}(Y^n)} & \leq  C(W,\pmax,\pavg). \label{eq:lastIneq}
\end{align*}
This completes the proof of \eqref{eq:StrongConverseProperty}. Thus, \eqref{eq:Strong} holds, and the strong converse property is proved for the DTPC $W$ under peak and average power constraints $\pmax$ and $\pavg$, respectively. That means that the RI capacity of the DTPC $C_{ID}(W,\pmax,\pavg)$ coincides with its transmission capacity $C(W,\pmax,\pavg)$.

\section{Conclusions} \label{sec: conclusion}

In our work, we perform an information-theoretic analysis for event-triggered molecular communication. In the previous considerations known to us, the model was considered with deterministic encoding. We now assume in our analysis that local randomness is available at the sender. Instead of $n^{nR_1}$ messages in the deterministic case, we get $2^{2^{nR_2}}$ in the randomised scenario. Furthermore, we derive an exact formula for the RI capacity, whereas in the DI case, only upper and lower bounds have been established so far. Of course, one can ask how realistic it is to have randomness available at the sender. One possible scenario would be a sender outside the body (for example in a bracelet). Another scenario would be to gain randomization through resources such as feedback.
Finally, we examine the scenario where the channel depends on a random state and derive the corresponding capacity.
In future studies, we should investigate deterministic identification with resources (as indicated above) that allow randomized encoding. Additionally, a fundamental objective is to analyze encoding techniques with finite block lengths for this model. Our work marks the initial step in undertaking this analysis.

\section*{Acknowledgments}
The authors acknowledge the financial support by the Federal Ministry of Education and Research
of Germany (BMBF) in the programme of “Souverän. Digital. Vernetzt.”. Joint project 6G-life, project identification number: 16KISK002.
H. Boche and W. Labidi were further supported in part by the BMBF within the national initiative on Post Shannon Communication (NewCom) under Grant 16KIS1003K. C.\ Deppe was further supported in part by the BMBF within the national initiative on Post Shannon Communication (NewCom) under Grant 16KIS1005. C. Deppe was also supported by the DFG within the project DE1915/2-1.

\bibliographystyle{IEEEtran}
\bibliography{definitions,references}
\clearpage
\newpage

\IEEEtriggeratref{4}



\end{document}

%% file: figures/System.tex
\tikzstyle{farbverlauf} = [ top color=white, bottom color=white!80!gray]
\tikzstyle{block} = [draw,top color=white, bottom color=white!80!gray, rectangle, rounded corners,
minimum height=2em, minimum width=2.5em]
\tikzstyle{input} = [coordinate]
\tikzstyle{sum} = [draw, circle,inner sep=0pt, minimum size=2mm,  thick]
\scalebox{.9}{
\tikzstyle{arrow}=[draw,->] 
\begin{tikzpicture}[auto, node distance=2cm,>=latex']
\node[] (M) {$i \in \mathcal{N}$};
\node[block,right=.5cm of M] (enc) {Enc};
\node[block, right=.7cm of enc] (channel) {$W_S$};
\node[block,below=.4cm of channel](state){$P_S$};
\node[block, right=.7cm of channel] (dec) {Dec};
\node[right=.5cm of dec] (Mhat) {\begin{tabular}{c} Is ${i}^\prime$ sent? \\ Yes or No? \end{tabular}};
\node[input,right=.5cm of channel] (t1) {};
\node[input,above=1cm of t1] (t2) {};
\draw[-{Latex[length=1.5mm, width=1.5mm]},thick] (M) -- (enc);
\draw[-{Latex[length=1.5mm, width=1.5mm]},thick] (enc) --node[above]{ $X^n$} (channel);
\draw[-{Latex[length=1.5mm, width=1.5mm]},thick] (channel) --node[above]{$Y^n$} (dec);
\draw[-{Latex[length=1.5mm, width=1.5mm]},thick] (dec) -- (Mhat);
\draw[-{Latex[length=1.5mm, width=1.5mm]},thick] (state)--(channel);
\end{tikzpicture}}

%% file: main.bbl
\begin{thebibliography}{10}
\providecommand{\url}[1]{#1}
\csname url@samestyle\endcsname
\providecommand{\newblock}{\relax}
\providecommand{\bibinfo}[2]{#2}
\providecommand{\BIBentrySTDinterwordspacing}{\spaceskip=0pt\relax}
\providecommand{\BIBentryALTinterwordstretchfactor}{4}
\providecommand{\BIBentryALTinterwordspacing}{\spaceskip=\fontdimen2\font plus
\BIBentryALTinterwordstretchfactor\fontdimen3\font minus
  \fontdimen4\font\relax}
\providecommand{\BIBforeignlanguage}[2]{{%
\expandafter\ifx\csname l@#1\endcsname\relax
\typeout{** WARNING: IEEEtran.bst: No hyphenation pattern has been}%
\typeout{** loaded for the language `#1'. Using the pattern for}%
\typeout{** the default language instead.}%
\else
\language=\csname l@#1\endcsname
\fi
#2}}
\providecommand{\BIBdecl}{\relax}
\BIBdecl

\bibitem{nakano2012}
T.~Nakano, M.~J. Moore, F.~Wei, A.~V. Vasilakos, and J.~Shuai, ``Molecular
  communication and networking: Opportunities and challenges,'' \emph{IEEE
  transactions on nanobioscience}, vol.~11, no.~2, pp. 135--148, 2012.

\bibitem{farsad2016}
N.~Farsad, H.~B. Yilmaz, A.~Eckford, C.-B. Chae, and W.~Guo, ``A comprehensive
  survey of recent advancements in molecular communication,'' \emph{IEEE
  Communications Surveys \& Tutorials}, vol.~18, no.~3, pp. 1887--1919, 2016.

\bibitem{muller2004challenges}
R.~H. Muller and C.~M. Keck, ``Challenges and solutions for the delivery of
  biotech drugs--a review of drug nanocrystal technology and lipid
  nanoparticles,'' \emph{Journal of biotechnology}, vol. 113, no. 1-3, pp.
  151--170, 2004.

\bibitem{jain1999transport}
R.~K. Jain, ``Transport of molecules, particles, and cells in solid tumors,''
  \emph{Annual review of biomedical engineering}, vol.~1, no.~1, pp. 241--263,
  1999.

\bibitem{nakano2014molecular}
T.~Nakano, T.~Suda, Y.~Okaie, M.~J. Moore, and A.~V. Vasilakos, ``Molecular
  communication among biological nanomachines: A layered architecture and
  research issues,'' \emph{IEEE transactions on nanobioscience}, vol.~13,
  no.~3, pp. 169--197, 2014.

\bibitem{AbnormalityDetection1}
S.~Ghavami and F.~Lahouti, ``Abnormality detection in correlated gaussian
  molecular nano-networks: Design and analysis,'' \emph{IEEE Transactions on
  NanoBioscience}, vol.~16, no.~3, pp. 189--202, 2017.

\bibitem{AbnormalityDetection2}
R.~Mosayebi, V.~Jamali, N.~Ghoroghchian, R.~Schober, M.~Nasiri-Kenari, and
  M.~Mehrabi, ``Cooperative abnormality detection via diffusive molecular
  communications,'' \emph{IEEE Transactions on NanoBioscience}, vol.~16, no.~8,
  pp. 828--842, 2017.

\bibitem{AbnormalityDetection3}
N.~Ghoroghchian, M.~Mirmohseni, and M.~Nasiri-Kenari, ``Abnormality detection
  and monitoring in multi-sensor molecular communication,'' \emph{IEEE
  Transactions on Molecular, Biological and Multi-Scale Communications},
  vol.~5, no.~2, pp. 68--83, 2019.

\bibitem{nano1}
N.~Varshney, A.~Patel, Y.~Deng, W.~Haselmayr, P.~K. Varshney, and
  A.~Nallanathan, ``Abnormality detection inside blood vessels with mobile
  nanomachines,'' \emph{IEEE Transactions on Molecular, Biological and
  Multi-Scale Communications}, vol.~4, no.~3, pp. 189--194, 2018.

\bibitem{nano2}
T.~Nakano, Y.~Okaie, S.~Kobayashi, T.~Koujin, C.-H. Chan, Y.-H. Hsu, T.~Obuchi,
  T.~Hara, Y.~Hiraoka, and T.~Haraguchi, ``Performance evaluation of
  leader–follower-based mobile molecular communication networks for target
  detection applications,'' \emph{IEEE Transactions on Communications},
  vol.~65, no.~2, pp. 663--676, 2017.

\bibitem{Mathfoundations}
Y.-P. {Hsieh} and P.-C. {Yeh}, ``{Mathematical Foundations for Information
  Theory in Diffusion-Based Molecular Communications},'' \emph{arXiv e-prints},
  p. arXiv:1311.4431, Nov. 2013.

\bibitem{mathfoundation2}
M.~Pierobon and I.~F. Akyildiz, ``Capacity of a diffusion-based molecular
  communication system with channel memory and molecular noise,'' \emph{IEEE
  Transactions on Information Theory}, vol.~59, no.~2, pp. 942--954, 2013.

\bibitem{capacity2}
N.~Farsad, Y.~Murin, A.~W. Eckford, and A.~Goldsmith, ``Capacity limits of
  diffusion-based molecular timing channels with finite particle lifetime,''
  \emph{IEEE Transactions on Molecular, Biological and Multi-Scale
  Communications}, vol.~4, no.~2, pp. 88--106, 2018.

\bibitem{Idchannels}
R.~{Ahlswede} and G.~{Dueck}, ``Identification via {C}hannels,'' \emph{IEEE
  Transactions on Information Theory}, vol.~35, no.~1, pp. 15--29, 1989.

\bibitem{Idfeedback}
------, ``Identification in the presence of feedback-a discovery of new
  capacity formulas,'' \emph{IEEE Transactions on Information Theory}, vol.~35,
  no.~1, pp. 30--36, 1989.

\bibitem{ja1985identification}
J.~J. Ja, ``Identification is easier than decoding,'' in \emph{26th Annual
  Symposium on Foundations of Computer Science (sfcs 1985)}.\hskip 1em plus
  0.5em minus 0.4em\relax IEEE, 1985, pp. 43--50.

\bibitem{yao1979some}
A.~C.-C. Yao, ``Some complexity questions related to distributive computing
  (preliminary report),'' in \emph{Proceedings of the eleventh annual ACM
  symposium on Theory of computing}, 1979, pp. 209--213.

\bibitem{AhlDueck}
R.~{Ahlswede} and G.~{Dueck}, ``Identification via channels,'' \emph{IEEE
  Transactions on Information Theory}, vol.~35, no.~1, pp. 15--29, Jan 1989.

\bibitem{verdu1993explicit}
S.~Verdu and V.~K. Wei, ``Explicit construction of optimal constant-weight
  codes for identification via channels,'' \emph{IEEE Transactions on
  Information Theory}, vol.~39, no.~1, pp. 30--36, 1993.

\bibitem{derebeyouglu2020performance}
S.~Derebeyo{\u{g}}lu, C.~Deppe, and R.~Ferrara, ``Performance analysis of
  identification codes,'' \emph{Entropy}, vol.~22, no.~10, p. 1067, 2020.

\bibitem{von2023identification}
C.~Von~Lengerke, A.~Hefele, J.~A. Cabrera, O.~Kosut, M.~Reisslein, and F.~H.
  Fitzek, ``Identification codes: A topical review with design guidelines for
  practical systems,'' \emph{IEEE Access}, 2023.

\bibitem{SPBD21pc}
M.~J. Salariseddigh, U.~Pereg, H.~Boche, and C.~Deppe, ``Deterministic
  identification over channels with power constraints,'' \emph{IEEE
  Transactions on Information Theory}, vol.~68, no.~1, pp. 1--24, 2021.

\bibitem{SPBD21f}
------, ``Deterministic identification over fading channels,'' in \emph{2020
  IEEE Information Theory Workshop (ITW)}.\hskip 1em plus 0.5em minus
  0.4em\relax IEEE, 2021, pp. 1--5.

\bibitem{SPBD21p}
M.~J. Salariseddigh, U.~Pereg, H.~Boche, C.~Deppe, and R.~Schober,
  ``Deterministic identification over poisson channels,'' in \emph{2021 IEEE
  Globecom: Workshop on Channel Coding beyond 5G}.\hskip 1em plus 0.5em minus
  0.4em\relax IEEE, Dec 2021.

\bibitem{SV23isi}
M.~J. Salariseddigh, V.~Jamali, U.~Pereg, H.~Boche, C.~Deppe, and R.~Schober,
  ``Deterministic identification for mc isi-poisson channel,'' in \emph{ICC
  2023 - IEEE International Conference on Communications}, May 2023.

\bibitem{SPBD21}
M.~J. Salariseddigh, U.~Pereg, H.~Boche, and C.~Deppe, ``Deterministic
  identification over channels with power constraints,'' in \emph{ICC 2021 -
  IEEE International Conference on Communications}.\hskip 1em plus 0.5em minus
  0.4em\relax IEEE, Jun 2021.

\bibitem{jamali2019channel}
V.~Jamali, A.~Ahmadzadeh, W.~Wicke, A.~Noel, and R.~Schober, ``Channel modeling
  for diffusive molecular communication—a tutorial review,''
  \emph{Proceedings of the IEEE}, vol. 107, no.~7, pp. 1256--1301, 2019.

\bibitem{gohari2016information}
A.~Gohari, M.~Mirmohseni, and M.~Nasiri-Kenari, ``Information theory of
  molecular communication: Directions and challenges,'' \emph{IEEE Transactions
  on Molecular, Biological and Multi-Scale Communications}, vol.~2, no.~2, pp.
  120--142, 2016.

\bibitem{unterweger2018experimental}
H.~Unterweger, J.~Kirchner, W.~Wicke, A.~Ahmadzadeh, D.~Ahmed, V.~Jamali,
  C.~Alexiou, G.~Fischer, and R.~Schober, ``Experimental molecular
  communication testbed based on magnetic nanoparticles in duct flow,'' in
  \emph{2018 IEEE 19th International Workshop on Signal Processing Advances in
  Wireless Communications (SPAWC)}.\hskip 1em plus 0.5em minus 0.4em\relax
  IEEE, 2018, pp. 1--5.

\bibitem{shamaiCapacityAchieving}
S.~Shamai, ``Capacity of a pulse amplitude modulated direct detection photon
  channel,'' \emph{IEE Proceedings I (Communications, Speech and Vision)}, vol.
  137, pp. 424--430(6), December 1990.

\bibitem{BarlettaShamai}
A.~{Dytso}, L.~{Barletta}, and S.~{Shamai}, ``{Properties of the Support of the
  Capacity-Achieving Distribution of the Amplitude-Constrained Poisson Noise
  Channel},'' \emph{arXiv e-prints}, Apr. 2021.

\bibitem{lapidoth2008capacity}
A.~Lapidoth and S.~M. Moser, ``On the capacity of the discrete-time poisson
  channel,'' \emph{IEEE Transactions on Information Theory}, vol.~55, no.~1,
  pp. 303--322, 2008.

\bibitem{lapidothWang}
A.~Lapidoth, J.~H. Shapiro, V.~Venkatesan, and L.~Wang, ``The discrete-time
  poisson channel at low input powers,'' \emph{IEEE Transactions on Information
  Theory}, vol.~57, no.~6, pp. 3260--3272, 2011.

\bibitem{aminian}
G.~Aminian, H.~Arjmandi, A.~Gohari, M.~Nasiri-Kenari, and U.~Mitra, ``Capacity
  of diffusion-based molecular communication networks over lti-poisson
  channels,'' \emph{IEEE Transactions on Molecular, Biological and Multi-Scale
  Communications}, vol.~1, no.~2, pp. 188--201, 2015.

\bibitem{boche2023algorithmic}
H.~Boche, A.~Grigorescu, R.~F. Schaefer, and H.~V. Poor, ``Algorithmic
  computability of the capacity of gaussian channels with colored noise,''
  \emph{arXiv preprint arXiv:2305.02819}, 2023.

\bibitem{HanBook}
T.~S. {Han}, \emph{Information-Spectrum Methods in Information Theory}, ser.
  Stochastic Modelling and Applied Probability.\hskip 1em plus 0.5em minus
  0.4em\relax Springer-Verlag Berlin Heidelberg, 2014.

\bibitem{capacityAchievingDistribution}
J.~Cao, S.~Hranilovic, and J.~Chen, ``Capacity-achieving distributions for the
  discrete-time poisson channel—part i: General properties and numerical
  techniques,'' \emph{IEEE Transactions on Communications}, vol.~62, no.~1, pp.
  194--202, 2014.

\end{thebibliography}
